\documentclass[12pt,a4paper,oneside]{article}
\usepackage[centertags]{amsmath}%
\usepackage{color}
\usepackage{graphicx}
\usepackage{amsthm,enumerate}
\usepackage[psamsfonts]{amsfonts,amssymb}

\newtheorem{definition}{Definition}
\newtheorem{theorem}[definition]{Theorem}
\newtheorem{proposition}[definition]{Proposition}

\newtheorem{corollary}[definition]{Corollary}
\newtheorem{rem}[definition]{Remark}
\newenvironment{remark}{\begin{rem}  \rm }{\end{rem}}
\newtheorem{rems}[definition]{Remarks}
\newenvironment{remarks}{\begin{rems}  \rm }{\end{rems}}
\newtheorem{example}[definition]{Example}

\newcommand\unit{\hbox{\rm 1\kern-2.8truept l}} 
\newcommand\re{\Re \kern-1.4truept e}
\newcommand\im{\Im \kern-1.4truept m}
\newcommand\Ker{\rm{Ker}}
\newcommand\Ran{{\rm{Ran\,}}}

\newcommand\spanno{{\rm{span}}}

\newcommand\M{\mathcal{M}}
\def\rel3{\alpha}

\newcommand{\h}{\mathsf{h}} 
\newcommand{\kk}{\mathsf{k}} 
\newcommand{\mm}{\mathsf{m}}

\newcommand{\ff}{\mathsf{f}}
\newcommand{\Ss}{\mathsf{s}}

\newcommand{\T}{\mathcal{T}} 

\newcommand{\wT}{\widetilde{\T}}

\newcommand{\Ll}{\mathcal{L}} 

\newcommand{\nT}{\mathcal{N(\mathcal{T})}} 
\newcommand{\FT}{\mathcal{F(\mathcal{T})}} 

\newcommand{\B}{\mathcal{B}(\h) } 

\newcommand{\E}{\mathcal{E}} 
\newcommand{\C}{\mathbb{C}} 

\newcommand{\nm}{\|} 

\newcommand{\RT}{\mathcal{R}(\T_*)}
\newcommand{\scal}[2]{\langle{#1},{#2}\rangle} 
\newcommand{\ee}[2]{|{e_{#1}}\rangle\langle{e_{#2}}|} 
\newcommand{\kb}[2]{|{#1}\rangle\langle{#2}|}

\newcommand{\tr}[1]{{\rm tr }\left(#1\right)}

\newcommand{\mi}{{\mathrm{i}}}

\begin{document}

\title{Relationships between the decoherence-free algebra and the set of fixed points}

\author{F. Fagnola$^1$, E. Sasso$^2$,
V. Umanit\`a$^2$}

\date{}

\author{F. Fagnola$^1$, E. Sasso$^2$,
V. Umanit\`a$^2$}
\maketitle

\maketitle

\begin{abstract}
We show that, for a {Quantum Markov Semigroup (QMS)} with a faithful normal invariant state, the atomicity of the
decoherence-free subalgebra and environmental decoherence are equivalent. Moreover, we characterize the set of reversible
states and explicitly describe the relationship between the decoherence-free subalgebra
and the fixed point subalgebra for QMSs with the above equivalent properties.
\end{abstract}

\section{Introduction}
Starting from the fundamental papers of Gorini-Kossakowski-Sudharshan \cite{GoKoSu}
and Lindblad \cite{Lindb} the structure of uniformly continuous quantum Markov
semigroups (QMS) $\T=(\T_t)_{t\ge 0}$, or, in the physical terminology, quantum
dynamical semigroups, and their generators, has been the object of significant attention.

The increasing interest in mathematical modelling of decoherence, coherent quantum
computing and approach to equilibrium in open quantum systems continues to motivate
investigation on special features of QMS. Special attention is paid to subalgebras or subspaces where irreversibility and dissipation disappear (see \cite{AFR,ABFJ,BlOl,DFSU,DFR,LW,TV}
and the references therein). States leaving in such subspaces are promising candidates
for storing and manipulating quantum information.

The decoherence-free subalgebra, where completely positive maps $\T_t$ of the semigroup
act as {automorphisms}, and the set of fixed points, which is a subalgebra when there exists
a faithful invariant state, also allow us to gain insight into the structure of a QMS, its invariant states
and environment induced decoherence. Indeed, its structure as a von
Neumann algebra, has important consequences on the structure and the action of the whole
QMS. Recently, we showed in \cite{DFSU} that, if the decoherence-free subalgebra of
a {uniformly} continuous QMS is atomic, it induces a decomposition of the system into its noiseless
and purely dissipative parts, determining the structure of invariant states, as well as
decoherence-free subsystems and subspaces \cite{TV}. In particular, we provided a full
description of invariant states extending known ones in the finite dimensional case \cite{BN}.

In this paper we push further the analysis of {uniformly} continuous QMS with atomic
decoherence-free subalgebra and a faithful invariant state proving a number new
results we briefly list and outline below.
\begin{enumerate}
\item Environment induced decoherence (\cite{BlOl,CSU2,CSU-new}) holds if and only if
$\nT$ is atomic. In this case the decoherence-free subalgebra is generated by the set of
eigenvectors corresponding to modulus one eigenvalues of the completely positive maps $\T_t$,
namely, in an equivalent way, by the eigenvectors with purely imaginary eigenvalue of the
generator (Theorem \ref{equiv-NT-atom}).
\item The decoherence-free subalgebra and the set of so-called \emph{reversible states},
i.e. the linear space generated by eigevectors corresponding to modulus 1 eigenvalues
of predual maps $\T_{*t}$ are in the natural duality of a von Neumann algebra with its
predual (Theorem \ref{RT=NT_*}). Moreover, Theorems \ref{th:rev-stat} and \ref{th:form-eta}
explicitly describe the structure of reversible states.
\item We find a spectral characterization of the decomposition of the fixed point algebra (Theorem \ref{NT-atomic-FT}).
Moreover, the exact relationship between $\FT$ and $\nT$ (Theorems \ref{NT-atomic-FT}
and \ref{FT-atomic-NT}) is established in an explicit and constructive way allowing one
to find the structure of each one from the structure of the other.
\end{enumerate}

Loosely speaking one can say that, for QMSs with a faithful invariant state,
the same conclusions can be drawn replacing finite dimensionality of the system
Hilbert space by atomicity of the decoherence-free subalgebra.\\
Counterexamples (Examples \ref{ex:Mr} and \ref{NTRT})
show that, in general, the above conclusions may fail if
for QMSs without faithful normal invariant states.
\smallskip

{The above results, clarify then the relationships between the atomicity of the
decoherence-free subalgebra, environmental decoherence, ergodic decomposition of
the trace class operators, and the structure of fixed points.

In particular the first result implies that the decomposition induced by decoherence coincides
with the Jacobs-de-Leeuw-Glickeberg (JDG) splitting. Such decomposition was originally introduced
for weakly almost periodic semigroups and generalized to QMSs on von Neumann algebras in
\cite{Kummerer, Hellmich} at all. It is among the most useful tools in the study of the
asymptotic behavior of operators semigroups on Banach spaces or von Neumann algebras.
Indeed, it provides a decomposition of the space into the direct sum of the space generated
by eigenvectors of the semigroup associated with modulus $1$ eigenvalues, and the remaining
space, called stable, consisting of all vectors whose orbits have $0$ as a weak cluster point.
Under suitable conditions, we obtain the convergence to $0$  for each vector
belonging to the stable space.

On the other hand, the ergodic decomposition of trace class operators (which is a particular
case of the JDS splitting applied to the predual of $\T$), allows one, for instance, to
determine reversible subsytems by spectral calculus.} Determining reversible states, in
particular, is an important task in the study of irreversible (Markovian) dynamics because
these states retain their quantum features that are exploited in quantum
computation (see \cite{ABFJ,TV} and the references therein). {More precisely, reversible (or rotating) and invariant states of a quantum channel (acting on $M_n(\mathbb{C})$ for some $n>1$) allow to classify kinds of information that the process can preserve. When the space is finite-dimensional and there exists a faithful invariant state, the structure of these states can be easily found (see e.g. Lemma $6$  and Section $V$ in \cite{viola}, and Theorems $6.12$, $6.16$ in \cite{wolf}) through the decomposition of $\nT$ and the algebra of fixed points $\FT$  in \lq\lq block diagonal matrices\rq\rq, i.e. in their canonical form given by the structure theorem for matrix algebras (see Theorem $11.2$ in \cite{Take}). Since the same decomposition holds for atomic von Neumann algebras, in this paper we generalize these results to uniformly continuous QMSs acting on $\B$ with $\h$ infinite-dimensional.

The paper is organized as follows. In Section \ref{sect:df-subalg} we collect
some notation and known results on the structure of norm-continuous QMS with atomic
decoherence-free subalgebra and the structure of their invariant states.
In Section \ref{sect:at-dec}, after recalling some known results from \cite{CSU-new}
on the relationship between EID and Jacobs-de Leeuw-Glickeberg decomposition,
we prove {the main result of this paper: the equivalence between EID and atomicity of the decoherence free subalgebra.}
The predual algebra of the decoherence-free subalgebra is characterized
in Section \ref{sect:rev-stat} as the set of reversible states.
Finally, in Section \ref{sect:FT-NT}, we study the structure of the set
of fixed points of the semigroup and its relationships with the decomposition
of $\nT$ when this algebra is atomic.

\section{The structure of the decoherence-free algebra}
\label{sect:df-subalg}

Let $\h$ be a complex separable Hilbert space and let $\B$ the algebra of all bounded
operators on $\h$ with unit $\unit$. A QMS on $\B$  is a semigroup $\T=(\T_t)_{t\ge 0}$
of completely positive, identity preserving normal maps which is also weakly$^*$
continuous. In this paper we assume $\T$ uniformly continuous i.e.
\[
\lim_{t\to 0^+} \sup_{\Vert x\Vert \le 1}
\left\Vert \T_t(x)-x\right\Vert=0,
\]
so that there exists a linear bounded operator $\Ll$ on $\B$ such that $\T_t=e^{t\Ll}$.
The operator $\Ll$ is the generator of $\T$, and it can be represented in the
well-known (see \cite{Partha}) Gorini-Kossakowski-Sudarshan-Lindblad (GKSL) form as
\begin{equation}\label{eq:GKSL}
\Ll(x)=\mi [H,x]-\frac{1}{2}\sum_{\ell\geq 1}
\left(L_\ell^*L_\ell x-2L_\ell^*xL_\ell+xL_\ell^*L_\ell\right),
\end{equation}
where $H=H^*$ and $(L_\ell)_{\ell\geq 1}$ are bounded operators on $\h$ such that the series
$\sum_{\ell\geq 1}L^*_\ell L_\ell$ is strongly convergent  and $[\cdot,\cdot]$ denotes
the commutator $[x,y]=xy-yx$. The choice of operators $H$ and $(L_\ell)_{\ell\geq 1}$
is not unique, but this will not create any inconvenience in this paper. More precisely,
we have the following characterization (see \cite{Partha}, Proposition $30.14$ and the
discussion below the proof of Theorem $30.16$).

\begin{theorem} \label{GKSL}
Let $\Ll$ be the generator of a uniformly continuous QMS on $\B$. Then there exist a
bounded selfadjoint operator $H$ and a sequence $(L_\ell)_{\ell\geq 1}$ of elements in
${\mathcal{B}}(\h)$ such that:
\begin{enumerate}
\item $\sum_{\ell\geq 1}L^*_\ell L_\ell$ is strongly convergent,
\item if $\sum_{\ell\geq 0}\vert c_\ell\vert^2<\infty$ and $c_0 \unit
+\sum_{\ell\geq 1}c_\ell L_\ell=0$
for scalars $(c_\ell)_{\ell\geq 0}$ then $c_\ell=0$ for every $\ell\geq 0$,
\item $\mathcal{L}(x)=\mi [H,x]
      -\frac{1}{2}\sum_{\ell\geq 1}
      \left(L^*_\ell L_\ell x-2L^*_\ell xL_\ell+xL^*_\ell L_\ell\right)$
for all $x\in{\mathcal{B}}(\h)$.
\end{enumerate}
\end{theorem}

We recall that, for an arbitrary von Neumann algebra $\mathcal{M}$, its predual space
$\mathcal{M}_*$ is the space of $w^*$-continuous functionals on $\mathcal{M}$ (said \emph{normal}).
It is a well-known fact that for all $\omega\in\mathcal{M}_*$ there exists $\rho\in\mathfrak{I}(\h)$,
the space of trace-class operators, such that  $\omega(x)=\tr{\rho x}$ for all $x\in\mathcal{M}$.
In particular, if $\omega$ is a positive and $||\omega||=1$, it is called \emph{state}, and $\rho$
is positive with $\tr{\rho}=1$, i.e. $\rho$ is a \emph{density}.\\
If $\mathcal{M}=\B$, every normal state $\omega$ has a unique density $\rho$. Therefore, in this case, we can identify them. Finally, $\rho$ is \emph{faithful} if $\tr{\rho x}=0$ for a positive $x\in\B$ implies $x=0$ (see \cite{Take}, Definition 9.4).\smallskip

Given a $w^*$-continuous operator $\mathcal{S}:\mathcal{M}\to\mathcal{M}$, we can define its \emph{predual map}  $\mathcal{S}_*:\mathcal{M}_*\to\mathcal{M}_*$ as $\mathcal{S}_*(\omega)=\omega\circ\mathcal{S}$.\\
In particular, for $\mathcal{M}=\B$, by considering the predual map of every $\T_t$, we obtain the \emph{predual semigroup} $\T_*=(\T_{*t})_t$ satisfying
\[
\tr{\T_{*t}(\rho)x}=\tr{\rho\T_{*t}(x)}\qquad\forall\,\rho\in\mathfrak{I}(\h), \ x\in\B.
\]

The \emph{decoherence-free (DF) subalgebra} of $\T$ is defined by
\begin{equation}\label{eq:NT-def}
\nT=\{x\in\B\,:\,\T_t(x^*x)=\T_t(x)^*\T_t(x),\ \T_t(xx^*)=\T_t(x)\T_t(x)^*\  \forall\,t\geq 0\}.
\end{equation}
It is a well known fact that $\nT$ is the biggest von Neumann subalgebra of $\B$
on which every $\T_t$ acts as a $*$-homomorphism (see e.g.
Evans\cite{Evans} Theorem 3.1). {Moreover, the following facts hold (see \cite{DFR}  Proposition 2.1).}

\begin{proposition}\label{prop-struct-NT}
Let $\T$ be a QMS on $\B$ and let ${\mathcal{N}}(\T)$ be the set defined by
(\ref{eq:NT-def}). Then
\begin{enumerate}
\item  $\nT$ is invariant with respect to every $\T_t$,
\item $\nT=\{\delta_H^{(n)}(L_k), \delta_H^{(n)}(L_k^*)\,:\,n\geq 0\}^\prime$, where $\delta_H(x):=[H,x]$,
\item $\T_t(x)=e^{\mi t H}xe^{-\mi tH}$ for all $x\in\nT$, $t\geq 0$,
\item if $\T$ possesses a faithful normal invariant state, then $\nT$ contains the set of fixed points $\FT=\{L_k, L_k^*, H\,:\,k\geq 1\}^\prime$.
\end{enumerate}
\end{proposition}
{In addition, if the QMS is uniformly continuous, its action on $\nT$ is bijective.
\begin{theorem}\label{nT-aut} If $\T$ is a uniformly continuous QMS, then $\nT$ is the
biggest von Neumann subalgebra on which every map $\T_t$ acts as a $*$-automorphism.
\end{theorem}
\begin{proof}
The restriction of every $\T_t$ to $\nT$ is clearly injective thanks to item $3$ of Proposition \ref{prop-struct-NT}. Now,
given $x\in\nT$ and $t>0$, we have to prove that $x=\T_t(y)$ for some $y\in\nT$. \\
First of all note that, since the QMS is norm continuous, it can be extended to norm continuous
\emph{group} $(\T_t)_{-\infty < t < +\infty}$ of normal maps on $\B$, and, by analyticity
in $t$, $\T_{-t}(z)\in\nT$ for all $t>0$
and $z\in\nT$, and formula $\T_{-t}(z)=\hbox{\rm e}^{-\mi tH}z\,\hbox{\rm e}^{\mi tH}$ holds.
\end{proof}}
\smallskip

As we said in the introduction, we will study the relationships between the structure of $\nT$
and other problems in the theory of uniformly continuous QMSs in which
the \emph{atomicity} of $\nT$ plays a key role.

First of all, as shown in \cite{DFSU}, the structure of the decoherence-free subalgebra $\nT$ gives
information on the whole QMS. Let us recall these results. Assume that $\nT$ is an atomic
algebra, that is, there exists  an (at most countable) family $(p_i)_{i\in I}$ of mutually
orthogonal non-zero projections, which are minimal projections in the center of $\nT$,
such that $\sum_{i\in I} p_i =\unit$ and each von Neumann algebra $ p_i \nT p_i  $ is
a type I factor. In that case, the subalgebra $\nT$ can be decomposed as
\begin{equation}\label{eq:NT-decomp}
\nT = \oplus_{i\in I} p_i \nT p_i\,.
\end{equation}
The properties of the projections $p_i$ imply their invariance under the semigroup
and more generally the $\T_t$-invariance of each factor $p_i\nT p_i$. Moreover,
each $p_i\nT p_i$ is a type I factor acting on the Hilbert
space $p_i\h$; thus, there exist two countable sequences
of Hilbert spaces $(\kk_i)_{i\in I}$, $(\mm_i)_{i\in I}$,
 and unitary operators $U_i:p_i\h\to\kk_i\otimes\mm_i$ such that
\begin{equation}\label{NTi}
U_ip_i\nT p_iU_i^*
=\mathcal{B}(\kk_i)\otimes\unit_{\mm_i},\qquad
U_ip_i\B p_iU_i^*=\mathcal{B}(\kk_i\otimes\mm_i).
\end{equation}
Therefore, defining $U=\oplus_{i\in I} U_i$, we obtain a unitary
operator $U:\h\to \oplus_{i\in I}(\kk_i\otimes \mm_i)$ such that
\begin{equation}\label{eq:UnT}
U\nT U^*=\oplus_{i\in I}\left(\mathcal{B}(\kk_i)\otimes\unit_{\mm_i}\right).
\end{equation}

As a consequence, we find the following result.

\begin{theorem}\label{th:main}
{$\nT$ is an atomic algebra if and only if there exist two countable sequences of
separable Hilbert spaces $(\kk_i)_{i\in I}$, $(\mm_i)_{i\in I}$ such that
$\h=\oplus_{i\in I}(\kk_i\otimes \mm_i)$ (up to a unitary operator) and
$\nT=\oplus_{i\in I}\left(\mathcal{B}(\kk_i)\otimes\unit_{\mm_i}\right)$ (up to an
isometric isomorphism).}
In this case:
\begin{enumerate}
\item for every GSKL representation of $\Ll$ by means of operators
$H,(L_\ell)_{\ell\ge 1}$, we have
\[
L_\ell =\oplus_{i\in I}
\left(\unit_{\kk_i}\otimes M_\ell^{(i)}\right)
\]
for a collection $(M_\ell^{(i)})_{\ell\geq 1}$ of operators in $\mathcal{B}(\mm_i)$,
such that the series $\sum_{\ell\ge 1}M_\ell^{(i)*}M_\ell^{(i)}$
strongly convergent for all $i\in I$, and
\[
H=\oplus_{i\in I}\left(K_i\otimes\unit_{\mm_i}
+\unit_{\kk_i}\otimes M_0^{(i)}\right)
\]
for  self-adjoint operators $K_i\in\mathcal{B}(\kk_i)$
and $M_0^{(i)}\in\mathcal{B}(\mm_i)$, $i\in I$,
\item  {defining on $\mathcal{B}(\mm_i)$ the GKSL generator $\Ll^{\mm_i}$ associated with operators $\{M_0^{(i)}, M_\ell^{(i)})\,:\,\ell^{(i)}\geq 1\}$, we have $$\T_t(x_i\otimes y_i)=\hbox{\rm e}^{\mi t K_i }x_i \hbox{\rm e}^{-\mi t K_i }\otimes\T^{\mm_i}(y_i)$$
for all $t\geq 0$, $x_i\in\mathcal{B}(\kk_i)$ and $y_i\in\mathcal{B}(\mm_i)$, where $\T^{\mm_i}$ is the QMS generated by $\Ll^{\mm_i}$, }
\item if there exists a faithful normal invariant state, then the QMS $\T^{\mm_i}$ is irreducible, possesses a unique invariant state $\tau_{\mm_i}$ which is also faithful, and we have $\mathcal{F}(\T^{\mm_i})=\mathcal{N}(\T^{\mm_i})=\mathbb{C}\unit_{\mm_i}$. Moreover, for all $i\in I$, $K_i$ has pure point spectrum.
\end{enumerate}
\end{theorem}

\begin{proof} The proof of the necessary condition is given in Theorems 3.2 and 4.1,
Proposition 4.3 and Lemma 4.2 in \cite{DFSU}. Conversely, given two countable sequences
of Hilbert spaces $(\kk_i)_i$, $(\mm_i)_i$, such that $\h=\oplus_{i\in I}(\kk_i\otimes \mm_i)$
up to a unitary operator, and $\nT=\oplus_{i\in I}\left(\mathcal{B}(\kk_i)\otimes\unit_{\mm_i}\right)$ (up to the corresponding isometric isomorphism), set $p_i$ the orthogonal projection onto $\kk_i\otimes\mm_i$. Then we obtain a family $(p_i)_{i\in I}$ of mutually orthogonal non-zero
projections, which are minimal projections in the center of $\nT$, such that
$\sum_{i\in I} p_i =\unit$ and each von Neumann algebra $ p_i \nT p_i=\mathcal{B}(\kk_i)\otimes\unit_{\mm_i} $ is a type I factor.

\end{proof}

Now, we recall a characterization of the atomicity of a von Neumann algebra in terms of the
existence of a normal conditional expectation, i.e. a weakly*-continuous norm one projection
(see \cite{Tom3} Theorem $5$).\\
To this end we will use that, given $x\in \h$ and $\kk$ complex separable Hilbert spaces and
$\sigma$ a normal state on $\kk$, there always exists (see e.g. Exercise $16.10$ in \cite{Partha})
a normal completely positive linear map $\E_\sigma:\mathcal{B}(\h\otimes\kk)\to\mathcal{B}(\h)$ satisfying
\begin{equation}\label{def-attesa-cond}
\tr{\E_\sigma(X)\eta}=\tr{X(\eta\otimes\sigma)}\qquad\forall\,X\in\mathcal{B}(\h\otimes\kk),\ \eta\in\mathfrak{I}(\h).
\end{equation}
Since $\E_\sigma$ is positive and $\E_\sigma\circ\E_\sigma=\E_\sigma$, it is a normal conditional expectation called \emph{$\sigma$-conditional expectation}.

\begin{theorem}[Tomiyama]\label{Tomi}
Let $\M$ be a von Neumann algebra acting on the Hilbert space $\h$. Then $\M$ is atomic if and
only if $\M$ is the image of a normal conditional expectation $\E:\B\to\M$.
\end{theorem}
\begin{proof} If there exists a normal conditional expectation $\E:\B\to\M$ onto $\M$, then $\M$
is atomic by Proposition $3$ and Lemma $5$ in \cite{Hayes}, being $\B$ atomic and semifinite.

On the other hand, if $\M$ is atomic, let $(p_i)_i$ be a sequence of orthogonal projections in
the center of $\M$ such that $p_i\M p_i$ is a type $I$ factor; moreover, assume $p_i\M p_i\simeq\mathcal{B}(\kk_i)\otimes\unit_{\mm_i}$ with $(\kk_i)_i$ and $(\mm_i)_i$ be two sequences
of complex and separable Hilbert spaces such that $\h\simeq\oplus_i\left(\kk_i\otimes\mm_i\right)$.
Set $(\sigma_i)_i$ a sequence of normal states on $(\mm_i)_i$, define the normal conditional expectations $\pi_i:\mathcal{B}(\kk_i\otimes\mm_i)\to\mathcal{B}(\kk_i)\otimes\unit_{\mm_i}$ as $\pi_i(x)=\E_{\sigma_i}(x)\otimes\unit_{\mm_i}$ for all $x\in\mathcal{B}(\kk_i\otimes\mm_i)$. So,
$$
\pi(x):=\sum_i\pi(x_{ii})\quad
\mbox{for $x=(x_{ij})_{ij}\in\mathcal{B}(\oplus_i\left(\kk_i\otimes\mm_i\right))\simeq\B$}
$$
gives a normal conditional expectation onto $\M$.
\end{proof}

\begin{corollary}\label{Tomi+cond-ex}
Let $\M$ be an atomic von Neumann algebra acting on $\h$ and let $\mathcal{N}\subseteq\M$
be a von Neumann subalgebra. If there exists a normal conditional expectation
$\E:\M\to\mathcal{N}$ onto $\mathcal{N}$, then $\mathcal{N}$ is atomic.
\end{corollary}
\begin{proof} By Tomiyama's Theorem we know that, since $\M$ is atomic, it is the image of
a normal conditional expectation $\mathcal{F}:\B\to\M$ (see also the proof of Theorem $5$ in \cite{Tom3}). Therefore, the map $\E\circ\mathcal{F}:\B\to\mathcal{N}$ is a normal conditional expectation onto $\mathcal{N}$.
Indeed, since $\mathcal{N}=\Ran\E$ is contained in $\M=\Ran\mathcal{F}$, for $x\in\B$ we have
\[
(\E\circ\mathcal{F})(\E\circ\mathcal{F})x=\E^2(\mathcal{F}(x))=(\E\circ\mathcal{F})(x),
\]
i.e. $\E\circ\mathcal{F}$ is a projection. Therefore, $\nm\E\circ\mathcal{F}\nm\geq 1$.
On the other hand, since $\E$ and $\mathcal{F}$ are norm one operators, we clearly obtain
$\nm\E\circ\mathcal{F}\nm=1$. The normality of $\E\circ\mathcal{F}$ is evident, and so we
can conclude that the algebra $\mathcal{N}$ is atomic by Tomiyama Theorem.
\end{proof}
\begin{remark}\rm
Theorem \ref{Tomi} is a simplified version of Theorem $5$ in \cite{Tom3}, given in terms of atomicity
of the subalgebra $\M$. Moreover, Corollary \ref{Tomi+cond-ex} generalizes to atomic algebras
one implication of the same theorem. In particular, we give easier proofs of these results.
\end{remark}

In the following we \emph{\underline{assume the existence of a faithful normal invariant} \underline{state}}; note that, in general, this condition is not necessary for the
decoherence-free subalgebra to be atomic. This is always the case for any QMS acting
on a finite dimensional algebra. However, we show the following example that will be useful
later.
\begin{example}\label{ex:atom-no-sif}\rm
Let $\h=\mathbb{C}^3$ with the canonical orthonormal basis $(e_i)_{i=1,2,3}$ and $\B=M_3(\mathbb{C})$. We consider the operator $\Ll$ on $M_3(\mathbb{C})$ given by
\[
\Ll(x)=\mi \omega\,[\ee{1}{1},x]
-\frac{1}{2}\left(\ee{3}{3}x-2\ee{3}{2}x\ee{2}{3}+x\ee{3}{3}\right)
\]
for all $x\in M_3(\mathbb{C})$, with $\omega\in\mathbb{R}$, $\omega\neq 0$. Clearly $\Ll$ is written in the GKSL form
with $H=\omega\ee{1}{1}$ and $L=\ee{2}{3}$, and so it generates a uniformly continuous QMS
$\T=(\T_t)_{t\geq 0}$ on $M_3(\mathbb{C})$.\\
An easy computation shows that any invariant functional has the form
$$a\ee{1}{1}+b\ee{2}{2}$$ for some $a,b\in\mathbb{C}$,
and so the semigroup has no faithful invariant states. \\
Since $[H,L]=0$, by item $2$ of Proposition \ref{prop-struct-NT} we have $\nT=\{L,L^*\}^\prime$,
so that an element $x\in\B$ belongs to $\nT$ if and only if
\[
\left\{\begin{array}{ll}\kb{xe_2}{e_3}=\kb{e_2}{x^*e_3}\\
\kb{xe_3}{e_2}=\kb{e_3}{x^*e_2}
\end{array}\right.,\quad\mbox{i.e.}\quad\left\{\begin{array}{lll}xe_2=x_{33}e_2\\
xe_3=x_{22}e_3\\
x_{31}=x_{21}=0
\end{array}\right..
\]
Therefore we get
\[
\nT=\left\{\,x_{11}\kb{e_1}{e_1}
+x_{22}\left(\kb{e_2}{e_2}+\kb{e_3}{e_3}\right)\,\mid\,x_{11}, x_{22}\in\mathbb{C}\,\right\},
\]
{i.e. $\nT$ is isometrically isomorphic to the atomic algebra $\mathbb{C}\oplus\mathbb{C}p$, where $p$ denotes the identity matrix in $M_2(\mathbb{C})$.}
\end{example}

\section{Atomicity of $\nT$ and decoherence}
\label{sect:at-dec}
In this section, we explore the relationships between the atomicity of $\nT$ and
the property of environmental decoherence, under the assumption of the existence of a
faithful normal invariant state $\rho$.
Following \cite{CSU-new}, we say that there is \emph{environment induced decoherence}
(EID) on the open system described by $\T$ if there exists a $\T_t$-invariant and
$*$-invariant weakly$^*$ closed subspace $\M_2$ of $\B$ such that:
 \begin{itemize}
 \item[(EID1)] $\B=\nT\oplus\M_2$ with $\M_2\not=\{0\}$,
 \item[(EID2)] $w^*-\lim_{t\to\infty}\T_t(x)=0$ for all $x\in\M_2$.
 \end{itemize}
Unfortunately, if EID holds and $\h$ is infinite-dimensional, it is not clear if the space $\M_2$ is uniquely determined. However,
 $\mathcal{M}_2$ is always contained in the $\T$-invariant and
$*$-invariant closed subspace
$$
\M_0=\left\{\, x\in\B\,:\,w^*-\lim_{t\to\infty}\T_t(x)=0\,\right\}.
$$
In \cite{DFSU} we showed that, if $\nT$ is atomic, then EID holds (see Theorem 5.1) and,
in particular, $\nT$ is the image of a normal conditional expectation $\E:\B\to\nT$ compatible
with the faithful state $\rho$ (i.e. $\rho\circ\E=\rho$) and such that {\begin{equation}\label{M2=Ker}{\Ker}\,\E=\M_2=\left\{\, x\in\B\,:\,\tr{\rho xy}=0\ \ \forall\, y\in\nT\,\right\}.\end{equation}}
(see Theorem $19$ in \cite{CSU-new}). In the following we will show that, if $\nT$ is atomic,
the decomposition unique, i.e. the only way to realize it, is taking $\M_2$ given by \eqref{M2=Ker}
(see Theorem \ref{equiv-NT-atom} and Remark \ref{univ-det}.$2$).\\ Moreover, in this case, we will
study the relationships of such a decomposition with another famous asymptotic splitting of $\B$,
called the \emph{Jacobs-de Leeuw-Glickberg splitting}: this comparison is very natural since
the decomposition $\B=\nT\oplus\M_2$ is clearly related too to the asymptotic properties of
the semigroup.

We recall that, since there exists $\rho$ faithful invariant, the Jacobs-de Leeuw-Glickberg
splitting holds (see e.g. Corollary $3.3$ and Proposition $3.3$ in \cite{Hellmich}) and is
given by $\B=\mathfrak{M}_r\oplus\mathfrak{M}_s$ with
\begin{align}
\mathfrak{M}_r&:=\overline{\spanno}^{w^*}\{x\in\B\,:\,
\T_t(x)=e^{\mi t\lambda}x\ \mbox{for some $\lambda\in\mathbb{R}$},\ \forall\,t\geq 0\}\\
\mathfrak{M}_s&:=\{x\in\B\,:\,0\in\overline{\{\T_t(x)\}}^{w^*}_{t\geq 0}\}.
\end{align}
{Moreover, in this case $\mathfrak{M}_r$ is a von Neumann algebra.}

The relationship between the decomposition induced by decoherence and the Jacobs-de Leeuw-Glickberg splitting is given by the following result
(see Proposition 31 in \cite{CSU-new}).
\begin{proposition}\label{EID=JDG}
If there exists a faithful normal invariant state $\rho$, then the following conditions
are equivalent:
\begin{enumerate}
\item EID holds with $\M_2=\M_0$ and the induced decomposition coincides with the Jacobs-de Leeuw-Glickberg splitting,
\item $\nT\cap\mathfrak{M}_s=\{0\}$,
\item $\nT=\mathfrak{M}_r$.
\end{enumerate}
Moreover, if one of the previous conditions holds, then $\nT$ is the image of a normal conditional
expectation $\E$ compatible with $\rho$ and such that ${\Ker}\,\E=\M_0=\mathfrak{M}_s$.
\end{proposition}
Clearly, if $\mathfrak{M}_r$ is not an algebra, it does not make sense to pose the problem to
understand if it coincides with $\nT$. In particular, this could happen when $\T$ has no
faithful invariant states, as the following example shows.

\begin{example}\label{ex:Mr}\rm
Let us consider the uniformly continuous QMS $\T$ on $M_3(\mathbb{C})$ defined in Example \ref{ex:atom-no-sif}. We have already seen that $\T$ does not posses faithful invariant states, and
\[
\nT=\left\{\,x_{11}\kb{e_1}{e_1}
+x_{22}\left(\kb{e_2}{e_2}+\kb{e_3}{e_3}\right)\,\mid\,x_{11}, x_{22}\in\mathbb{C}\,\right\}.
\]
We want now to find the space $\mathfrak{M}_r$, generated by eigenvectors of $\Ll$ corresponding to purely imaginary eigenvalues. Easy computations show that we have $\Ll(x)=\mi\lambda x$ for some $\lambda\in\mathbb{R}$ if and only if
\begin{align*}\mi\lambda\sum_{i,j=1}^3x_{ij}\ee{i}{j}&=\mi \omega\sum_{j=1}^3\left(x_{1j}\ee{1}{j}-x_{j1}\ee{j}{1}\right)\\ &-\frac{1}{2}\left(\sum_{j=1}^3x_{3j}\ee{3}{j}-2x_{22}\ee{3}{3}+\sum_{i=1}^3x_{i3}\ee{i}{3}\right),
\end{align*}
i.e., in the case when $\lambda=0$, $x_{ij}=0$ for $i\not=j$ and $x_{22}=x_{33}$, and,
 in the case $\lambda\not=0$, if and only if the following identities hold
\[
\begin{array}{ccc}
x_{11}=0,  & x_{22}=0, & x_{12}(\omega-\lambda)=0,\\
x_{13}\left(-\frac{1}{2}+\mi(\omega-\lambda)\right)=0
& x_{21}(\omega+\lambda)=0, & x_{23}(\frac{1}{2}+\mi\lambda)=0,\\
x_{31}\left(\frac{1}{2}+\mi(\omega+\lambda)\right)=0, &
x_{32}(\frac{1}{2}+\mi\lambda)=0, & x_{33}(1+\mi\lambda)=x_{22}.
\end{array}
\]
Since $\omega$ and $\lambda$ belong to $\mathbb{R}$ this is equivalent to have either $x=x_{12}\ee{1}{2}$ and $\lambda=\omega$, or $x=x_{21}\ee{2}{1}$ and $\lambda=-\omega$.
Therefore, we can conclude that
\[
\mathfrak{M}_r
=\left\{\left(\begin{array}{ccc}
x_{11} & x_{12} & 0 \\
x_{21} & x_{22} & 0 \\
0 & 0 & x_{22}
\end{array}\right)\,:\, x_{11}, x_{22}, x_{12}, x_{21}\in\mathbb{C}\right\}.
\]
In particular, $\mathfrak{M}_r$ is not an algebra and it is strictly bigger that $\nT$.
\end{example}
\begin{remarks}\rm $1.$ Note that if $\mathfrak{M}_r$ is contained in $\nT$, then it
is a $*$-algebra.\\
Indeed, if $\mathfrak{M}_r\subseteq\nT$, taken $x,y\in\mathfrak{M}_r$ such that
$\T_t(x)=e^{\mi\lambda t}x$ and $\T_t(y)=e^{\mi\mu t}y$ for some $\lambda,\mu\in\mathbb{R}$
and any $t$, by property $3$ in Proposition \ref{prop-struct-NT} we have
$$\T_t(x^*y)=\T_t(x)^*\T_t(y)=e^{\mi t(\mu-\lambda)}x^*y\qquad\forall\,t\geq 0.$$
As a consequence $x^*y$ belongs to $\mathfrak{M}_r$.

$2.$ If $\h$ is finite-dimensional, then also the opposite implication is true.\\
Indeed, if $\mathfrak{M}_r$ is a $*$-algebra, given $x\in\B$ such that
$\T_t(x)=e^{it\lambda}x$, $\lambda\in\mathbb{R}$, we have $\T_t(x^*)\T_t(x)=x^*x$.
Then, by the Schwarz inequality, $\T_t(x^*x)\geq x^*x$ for all $t\geq 0$. Set $\T^r_t:=\T_{t|_{\mathfrak{M}_r}}$. Since in this case $\T$ is a strongly continuous
semigroup, by definition of $\mathfrak{M}_r$ and by Corollary $2.9$, Chapter V, of
\cite{EN}, the strong operator closure of $\{\T^r_t\,:\, t\geq 0\}$ is a compact
topological group of operators in $\mathcal{B}(\mathfrak{M}_r)$. Hence, $(\T^r_t)^{-1}$
is the limit of some net $(\T^r_{t_\alpha})_\alpha$ and so $(\T^r_t)^{-1}$ is a positive
operator. Since $x^*x\in \mathfrak{M}_r$, for all $\alpha$ we have
$\T^r_{t_\alpha}(x^*x)\geq x^*x$, and so $(\T^r_t)^{-1}(x^*x)\geq x^*x$. On the other hand,
$$
(\T^r_t)^{-1}(\T^r_t(x^*x))=x^*x\geq(\T^r_t)^{-1}(x^*x).
$$
Therefore, $(\T^r_t)^{-1}(x^*x)=x^*x$ and this implies $\T_t(x^*x)=x^*x=\T_t(x)^*\T_t(x)$.
Similarly we can prove the equality $\T_t(xx^*)=xx^*=\T_t(x)\T_t(x)^*$, and so $x$ belongs to $\nT$.
\end{remarks}

Now we are able to prove one of the central results of this paper.
\begin{theorem}\label{equiv-NT-atom}
Assume that there exists a faithful normal invariant state $\rho$. Then $\nT$ is atomic
if and only if EID holds with $\nT=\mathfrak{M}_r$ and $\M_2=\M_0$.
\end{theorem}
\begin{proof}
If $\nT$ is atomic, then EID holds by Theorem $5.1$ in \cite{DFSU}. It remains to prove that $\nT=\mathfrak{M}_r$ and $\M_2=\M_0$.  The atomicity implies $\nT=\oplus_{i\in I}\left(\mathcal{B}(\kk_i)\otimes\unit_{\mm_i}\right)$
up to a unitary isomorphism. Let $x=\sum_{i\in I}(x_i\otimes\unit_{\mm_i})$
be in $\nT\cap\mathfrak{M}_s$, with $x_i\in\mathcal{B}(\kk_i)$ for every $i\in I$,
and assume $w^*-\lim_\alpha\T_{t_\alpha}(x)=0$. Given $u_i,v_i\in\kk_i$ and $\tau_i$
an arbitrary state on $\mm_i$, by Theorem \ref{th:main}
\begin{equation}\label{Tx}
\tr{(\kb{u_i}{v_i}\otimes\tau_i)\T_{t_\alpha}(x)}
=\scal{v_i}{e^{\mi t_\alpha K_i}x_ie^{-\mi t_\alpha K_i}u_i}.
\end{equation}
Choosing $u_i$ and $v_i$ such that $K_iu_i=\lambda_iu_i$ and $K_iv_i=\mu_iv_i$, $\lambda_i,\mu_i\in\mathbb{R}$, equation \eqref{Tx} becomes $$\tr{(\kb{u_i}{v_i}\otimes\tau_i)\T_{t_\alpha}(x)}=e^{\mi t_\alpha(\mu_i-\lambda_i)}\scal{v_i}{x_iu_i},$$ so that $\scal{v_i}{x_iu_i}=0$, i.e.
$x_i=0$ because the eigenvectors of $K_i$ from an orthonormal basis of $\kk_i$ (see
item $3$ in Theorem \ref{th:main}).
This proves the equality $\nT\cap\mathfrak{M}_s=\{0\}$.\\
So we can conclude thanks to item $3$ of Proposition \ref{EID=JDG}.

Conversely, if EID holds with $\nT=\mathfrak{M}_r$ and $\M_2=\M_0$, by Proposition \ref{EID=JDG}
there exists
a normal conditional expectation $\E:\B\to\nT$ onto $\nT$, compatible to $\rho$. Therefore,
$\nT$ is atomic thanks to Theorem \ref{Tomi}. \end{proof}

\begin{remarks}\label{univ-det}\rm As a consequence of Theorem \ref{equiv-NT-atom} and Proposition \ref{EID=JDG} the following facts hold:\smallskip

$1.$\, $\nT$ is atomic if and only if $\nT\cap\,\mathfrak{M}_s=\{0\}$, if and only if $\nT=\mathfrak{M}_r$, i.e $\nT$ is generated by eigenvectors of $\Ll$ corresponding to purely
imaginary eigenvalues. \\
Moreover, in this case we also have $\nT\cap\M_0=\{0\}$, being $\M_0\subseteq\mathfrak{M}_s$:
this means that, assuming $\nT$ atomic and the existence of a faithful invariant state,
the situation is similar to the finite-dimensional case, i.e $\nT$ does not contain operators
going to $0$ under the action of the semigroup.\smallskip

$2.$\,if $\nT$ is atomic and $\FT=\mathbb{C}\unit$, the semigroup satisfies the following properties given by non-commutative Perron-Frobenius Theorem (see e.g. Propositions $6.1$  and $6.2$ in \cite{batkai}, Theorem $2.5$ in \cite{albev-frobenius}):
\begin{itemize}
\item[-]  the peripheral point spectrum $\sigma_p(\T_t)\cap\mathbb{T}$ of each $\T_t$ is a subgroup of the circle group $\mathbb{T}$,
\item[-] given $t\geq 0$, each peripheral eigenvalue $\alpha$ of $\T_t$ is simple and we have $\sigma_p(\T_t)\cap\mathbb{T}=\alpha (\sigma_p(\T_t)\cap\mathbb{T})$,
\item[-] the restriction of $\rho$ to $\nT$ is a trace.
\end{itemize}
As a consequence, the peripheral point spectrum of each $\T_t$ is the cyclic group of all $h$-roots of unit for some $h\in\mathbb{N}$.\smallskip

$3.$\,  If $\nT$ is atomic, the decomposition induced by decoherence is \emph{uniquely determined}. This fact follows from Proposition $5$ in \cite{CSU-new}, since we have $\nT\cap\mathfrak{M}_0=\{0\}$.\smallskip

$4.$\, Note that Theorem \ref{equiv-NT-atom} does not exclude the possibility to have a QMS $\T$ displaying decoherence with $\nT$ a non-atomic type I algebra. Clearly, in this case, we will get $\nT\supsetneq\mathfrak{M}_r$ or $\M_2\subsetneq\M_0$.

\end{remarks}

\begin{remark}\rm
In \cite{LO} the authors prove that EID holds when the semigroup commutes with the modular group associated with a faithful normal invariant state. However, our result in Theorem \ref{equiv-NT-atom}
is stronger since we find the equivalence between EID and the atomicity of $\nT$, which is a weaker assumption of the commutation with the modular group.
In fact, it can be shown (see \cite{Take-sp-cond}, section $3$) that
commutation with the modular group implies atomicity of $\nT$. Moreover, it is not difficult to
find an example of a QMS on $\B$, with $\h$ finite dimensional which does not commute with the
modular group. Its decoherence-free subalgebra, as any finite dimensional von Neumann algebra,
will be atomic.
\end{remark}

\section{Structure of reversible states}
\label{sect:rev-stat}
In this section, assuming $\nT$ atomic and the existence of a faithful invariant state $\rho$,
we study the structure of \emph{reversible states}, i.e. states belonging to the vector space
\begin{align}\label{RT-T}
\RT:&=\overline{\spanno}\{\sigma\in\mathfrak{I}(\h)\,:\,\T_{*t}(\sigma)
=e^{\mi t\lambda}\sigma\ \mbox{for some $\lambda\in\mathbb{R},\ \forall\,t\geq 0$}\}\\
\label{RT-L}&=\overline{\spanno}\{\sigma\in\mathfrak{I}(\h)\,:\,
\Ll_{*}(\sigma)=\mi\lambda\,\sigma\ \mbox{for some $\lambda\in\mathbb{R}$}\}.
\end{align}
In particular
we will prove that $\RT$ is the predual of the decoherence-free algebra $\nT$.\smallskip

{To this end, we recall the following result which is a version of the Jacobi-De Leeuw-Glicksberg theorem for strongly continuous semigroup (see Propositions $3.1, 3.2$ in \cite{Kummerer} and Theorem $2.8$ in \cite{EN}).
\begin{theorem}\label{kummerer}
If there exists a normal density $\rho\in\mathfrak{I}(\h)$ satisfying
\begin{equation}\label{disug-inv}\tr{\rho\left(\T_t(x)^*\T_t(x)\right)}\leq\tr{x^*x}\quad\quad\forall\,x\in\B,\ t\geq 0,\end{equation}
then we can decompose $\mathfrak{I}(\h)$ as
\begin{equation}\label{dec-J}
\mathfrak{I}(\h)
=\RT\oplus\{\sigma\in\mathfrak{I}(\h)\,:\,0\in\overline{\{\T_{*t}(\sigma)\}}^w_{t\geq 0}\}.
\end{equation}
\end{theorem}
Since each faithful invariant state clearly fulfills \eqref{disug-inv}, we obtain the splitting given by equation \eqref{dec-J}.

On the other hand, denoting by $\,^\perp A$ the vector space $\{\sigma\in\mathfrak{I}(\h)\,:\, \tr{\sigma x}=0\ \forall\ x\in A\}$ for all subset $A$ of $\B$, the atomicity of $\nT$ ensures the following facts:}
\begin{enumerate}
\item[(F1).] $\B=\nT\oplus\M_0$ with $\nT=\mathfrak{M}_r=\Ran\E$ and $\M_0=\mathfrak{M}_s={\Ker}\E$, where $\E:\B\to\nT$ is a conditional expectation compatible with the faithful state $\rho$ (see Theorem \ref{equiv-NT-atom} and Proposition \ref{EID=JDG});
\item[(F2).] $\mathfrak{I}(\h)=\,^\perp\M_0\oplus\,^\perp\nT$ with $$\,^\perp\M_0=\Ran\E_*\simeq\nT_*,\qquad\,^\perp\nT={\Ker}\E_*\simeq\M_{2*}.$$ Moreover each $\T_{*t}$ acts as a surjective isometry on $\,^\perp\M_0$, and $\lim_t\T_{*t}(\sigma)=0$ for all $\sigma\in\,^\perp\nT$ (see Theorem $10$ in \cite{CSU-new}).
\end{enumerate}
As a consequence, every state $\omega\in\nT_*$ is represented by a unique density $\sigma$
in $\,^\perp\M_0$, and, in this case, we write $\omega=\omega_\sigma$ to mean that
$\omega(x)=\tr{\sigma x}$ for all $x\in\nT$. Therefore, if we denote by $\mathcal{S}=(\mathcal{S}_t)_{t\geq 0}$ the restriction of $\T$ to $\nT$, we have
$$
(\mathcal{S}_{*t}\omega_\sigma)(x)=\omega_\sigma(\T_t(x))
=\tr{\sigma\,e^{\mi t H}xe^{-\mi t H}}=\tr{\E_*(e^{-\mi t H}\sigma e^{\mi t H}) x}
$$
for all $x=\E(x)\in\nT$, concluding that $\mathcal{S}_{*t}\omega_\sigma$ is represented by
the density $\E_*(e^{-\mi t H}\sigma e^{\mi t H})\in\,^\perp\M_0$. In a equivalent way, we have
\begin{equation}\label{ev-perp-M_0}
\T_{*t}(\sigma)=\E_*(e^{-\mi t H}\sigma e^{\mi t H})\qquad\forall\,\sigma\in\,^\perp\M_0.
\end{equation}

\begin{theorem}\label{RT=NT_*}
If $\nT$ is atomic and there exists a faithful invariant state, then
$$
\RT=\,^\perp\M_0=\{\sigma\in\mathfrak{I}(\h)\,:\,\T_{*t}(\sigma)
=\E_*(e^{-\mi t H}\sigma\,e^{\mi t H})\ \,\forall\,t\geq 0\}\simeq\nT_*,
$$
for every Hamiltonian $H$ in a GKSL representation of the generator of $\T$.
\end{theorem}
\begin{proof}
The inclusion $\,^\perp\M_0\subseteq\{\sigma\in\mathfrak{I}(\h)\,:\,\T_{*t}(\sigma)
=\E_*(e^{-\mi t H}\sigma\,e^{\mi t H})\ \,\forall\,t\geq 0\}$ follows from the previous
discussion. On the other hand, if we have $\T_{*t}(\sigma)=\E_*(e^{-\mi t H}\sigma e^{\mi t H})$
for all $t\geq 0$, taking $t=0$ we get $\sigma=\E_*(\sigma)$, i.e. $\sigma$ belongs to $\,^\perp\M_0$.

Now, given $\sigma\in\RT$ such that $\T_{*t}(\sigma)=e^{\mi t\lambda}\sigma$
for all $t\geq 0$, $\lambda\in\mathbb{R}$, we have
\[
\tr{\sigma x}=\lim_{t\to\infty}\tr{\sigma x}=\lim_{t\to\infty}\tr{\T_{*t}(\sigma)
e^{-\mi t\lambda}x}=\lim_te^{-\mi t\lambda}\tr{\sigma\T_t(x)}=0
\]
for all $x\in\M_0$, so that $\sigma$ belongs to $\,^\perp\M_0$. This proves that $\RT$
is contained in $\,^\perp\M_0$.\\
In order to prove the opposite inclusion it is enough to show that $\,^\perp\nT$ contains $\{\sigma\in\mathfrak{I}(\h)\,:\,0\in\overline{\{\T_{*t}(\sigma)\}}^w_{t\geq 0}\}$, since
we have
\[
\mathfrak{I}(\h)=\RT\oplus\{\sigma\in\mathfrak{I}(\h)\,:\,
0\in\overline{\{\T_{*t}(\sigma)\}}^w_{t\geq 0}\}=\,^\perp\M_0\oplus\,^\perp\nT
\]
 by equation \eqref{dec-J}, item (F2) and Theorem \ref{equiv-NT-atom}.\\
So, let $\sigma\in\mathfrak{I}(\h)$ such that $0\in\overline{\{\T_{*t}(\sigma)\}}^w_{t\geq 0}$;
given $(t_\alpha)$ with $w-\lim_\alpha\T_{*t_\alpha}(\sigma)=0$ and $x\in\mathfrak{M}_r$
such that $\T_t(x)=e^{\mi t\lambda}x$ for some $\lambda\in\mathbb{R}$, we have
\[
\tr{\sigma x}=\lim_\alpha\tr{\sigma e^{-\mi t\lambda}\T_{t_\alpha}(x)}
=\lim_\alpha e^{-\mi t\lambda}\tr{\T_{*t_\alpha}(\sigma)x}=0.
\]
This means that $\sigma$ belongs to $\,^\perp\nT$ by Theorem \ref{equiv-NT-atom}.
\end{proof}

In general, when there does not exist a faithful invariant state $\RT$ could be different
from $\nT_*$, as we can see in Example \ref{NTRT}.
\begin{example}\label{NTRT}\rm
Let us consider a generic Quantum Markov Semigroup with $\mathbb{C}^3$, more precisely the uniformly continuous QMS generated by
\[
\Ll(x) = G^* x + \sum_{j=1,2}L_{3j}^* x L_{3j} + xG
\]
where
\begin{eqnarray*}
 G = \left(-\frac{\gamma_{33}}{2}
 + i \kappa_3 \right) |e_3\rangle \langle e_3|,
\qquad
L_{3j}=\sqrt{\gamma_{3j}}\,|e_j\rangle\langle e_3| \qquad \mbox{for } j=1,2,
\end{eqnarray*}
with $\kappa_3\in\mathbb{R}$, $\gamma_{3j}> 0$ for $j=1,2$, and $\gamma_{33}=-\gamma_{31}-\gamma_{32}$. We know that, the restriction of $\Ll$ to the diagonal matrices is the generator of a continuous time Markov chain $(X_t)_t$ with values in $\{1,2,3\}$. For more details see \cite{AFH,generic}.\\
Since $1$ and $2$ are absorbing states for $(X_t)_t$, and $3$ is a transient state, by Proposition $2$ in \cite{isolati} we know that any invariant state of $\T$ is supported on $\spanno\{e_1,e_2\}$. In particular, this implies there is no faithful invariant state. \\
Moreover, Theorem $8$ in \cite{isolati} gives $\nT=\mathbb{C}\unit$, since the absorbing states are accessible from $3$. As a consequence, $\nT_*=\C\unit$.\\
On the other hand, since $1$ is absorbing, the state $\ee{1}{1}$ is invariant, and so it belongs in particular to $\RT$. Therefore, we have $\RT\neq\nT_*$ .
\end{example}

We can now give the structure of reversible states when $\nT$ is a {type $I$ factor}.
\begin{theorem}\label{th:rev-stat}
Let $\T$ be a QMS on $\mathcal{B}(\kk\otimes\mm)$ with a faithful
invariant state $\rho$ and $\nT=\mathcal{B}(\kk)\otimes\unit_\mm$
and let $\tau_\mm$ be the unique invariant state of the partially
traced semigroup $\T^\mm$ defined in Theorem \ref{th:main} item {\rm 2.}
Then a state $\eta$ belongs to $\RT\simeq\nT_*$ if and only if
\begin{equation}\label{eq:form-eta}
\eta = \sigma \otimes \tau_\mm
\end{equation}
for some state $\sigma$ on $\mathcal{B}(\kk)$.\end{theorem}

\begin{proof} Let $(e_j)_{j\ge 1}$ be an orthonormal basis of eigenvectors of $K$ so that
$Ke_j = \kappa_je_j$ for some $\kappa_j \in \mathbb{R}$. Given a state $\eta$, we can write
\[
\eta=\sum_{j,k\ge 1} |e_j\rangle\langle e_k|
\otimes \eta_{jk}
\]
with $\eta_{jk}$ trace class operator on $\mm$, so that
\[
 \T_{*t}(\eta)
=\sum_{j,k\ge 1} \hbox{\rm e}^{\mi (\kappa_k-\kappa_j)t}
|e_j\rangle\langle e_k| \otimes \T^\mm_{*t}(\eta_{jk})
\]
by Theorem \ref{th:main}. Therefore, by the linear independence of operators
$|e_j\rangle\langle e_k|$, we have $\T_{*t}(\eta)=e^{\mi t\lambda}\eta$ for some
$\lambda\in\mathbb{R}$ if and only if
\[
\T^\mm_{*t}(\eta_{jk})=e^{\mi t(\lambda-\kappa_k+\kappa_j)}\eta_{jk}\qquad\forall\,j,k,
 \]
 i.e. if and only if each $\eta_{jk}$ belongs to $\mathcal{R}(\T^\mm_*)$.\\
Now, by item $3$ of Theorem \ref{th:main} we know that $\T^{\mm}$ has a unique (faithful)
invariant state $\tau_\mm$ and  $\mathcal{N}(\T^\mm)=\mathcal{F}(\T^{\mm})$; hence,
Theorem \ref{RT=NT_*} and Proposition \ref{FT-atom} give $\mathcal{R}(\T^\mm_*)\simeq\mathcal{N}(\T^\mm)_*=\mathcal{F}(\T^{\mm})_*
=\mathcal{F}(\T^{\mm}_*)=\spanno\{\tau_\mm\}$. As a consequence, we can conclude that
$\eta$ belongs to $\nT_*$ if and only if $\eta_{jk}=\tr{\eta_{jk}}\tau_\mm$ for all
$j,k\ge 1$, i.e. if and only if
\[
\eta = \sum_{j,k} \left(\tr{\eta_{jk}} |e_j\rangle\langle e_k|\right)
\otimes \tau_\mm=\sigma\otimes\tau_\mm
\]
with $\sigma:=\sum_{j,k} \tr{\eta_{jk}} |e_j\rangle\langle e_k|$.\\
This is enough to prove the statement since $\RT$ is the vector space generated
by eigenstates of $\T_{*t}$ corresponding to modulo $1$ eigenvalues.
\end{proof}

If $\nT$ is not a type I factor, but it is atomic, we can obtain a similar result using
that reversible states are \lq\lq block-diagonal\rq\rq, as the following proposition shows.
\begin{proposition}\label{eta-pi-pj} Assume $\nT$ atomic. Let $\T$ be a QMS with a faithful
invariant state and let $\nT$ as in \eqref{eq:NT-decomp} with $(p_i)_{i\in I}$ minimal
projections in the center of $\nT$. Then $p_i\sigma p_j=0$ for
all $i\not=j$ and for all reversible state $\sigma\in\RT$.
\end {proposition}
\begin{proof}
Let $\sigma\in\RT$ such that $\T_{*t}(\sigma)=e^{\mi t\lambda}\sigma$ for some
$\lambda\in\mathbb{R}$. Since $\RT=\,^\perp\,\M_0$, by (F1) we have
$\tr{\sigma x}=\tr{\sigma x_1}$ for all $\B\ni x=x_1+x_2$ with $x_1\in\nT$ and $x_2\in \M_0$.
Moreover, the property $\T_t(p_k)=p_k$ for all $k\in I$ and $t\geq 0$ gives
$$
\T_{*t}(p_i\sigma p_j)=p_i\T_{*t}(\sigma)p_j
=e^{\mi  t\lambda}\,p_i\sigma p_j\qquad\, \forall\ i,j\in I,
$$
so that each $p_i\sigma p_j$ belongs to $\RT\simeq\nT_*$. Therefore, since
$\tr{p_i\sigma p_jx}=\tr{\sigma p_j xp_i}=0$ for all $x\in\nT=\oplus_{i\in I}p_i\nT p_i$, $i\neq j$,
we obtain that $p_i\sigma p_j=0$ for all $i\neq j$.
\end{proof}
As a consequence, by Theorem \ref{th:main} we have the following characterization of
reversible states
\begin{theorem}\label{th:form-eta}
Assume $\nT$ atomic and suppose there exists a faithful $\T$- invariant
state.  Let $(p_i)_{i\in I}$, $(\kk_i)_{i\in I}$, $(\mm_i)_{i\in I}$ be as
in Theorem \ref{th:main}. A  state $\eta$ belongs to $\RT$ if and only if it can be
written in the form
\[
\eta=\sum_{i\in I}
\tr{\eta p_i} \sigma_i \otimes \tau_{\mm_i}
\]
where, for every $i\in I$,
\begin{enumerate}
\item $\tau_{\mm_i}$ is the unique $\T^{\mm_i}$-invariant state
which is also faithful,
\item $\sigma_i$ is a density on $\kk_i$.\end{enumerate}
\end{theorem}
We have thus derived the general form of reversible states starting from the structure of the
atomic decoherence-free algebra $\nT$. This is a well known fact for a completely positive and
unitary map (i.e. a channel) on a finite-dimensional space (see e.g. Theorem $6.16$ in \cite{wolf}
and section $V$ in \cite{viola}), but the proof of this result is not generalizable to the infinite dimensional case since it is based on a spectral decomposition of the channel
based on eigenvectors.

\section{Relationships with the structure of fixed points}
\label{sect:FT-NT}
In this section we investigate the structure of the set $\FT$ of fixed points of the
semigroup and its relationships with the decomposition of $\nT$ given in the previous
section.

{First of all we prove the atomicity of $\FT$ and
relate this algebra with the space of invariant states. Really, the reader can find the proof of these results in \cite{Frigerio-staz}. We report them for sake of completeness.}

\begin{proposition}\label{FT-atom}
If there exists a faithful normal invariant state, then $\FT$ is an atomic algebra and
$\FT_*$ is isomorphic to the space $\mathcal{F}(\T_*)$ of normal invariant functionals.
\end{proposition}
\begin{proof}
Since there exists a faithful invariant state, by Theorem $2.1$ in \cite{Frigerio-staz}
and in \cite{FV82} $\FT$ it is the image of a normal conditional expectation $\E:\B\to\FT$
given by
\begin{equation}\label{E-su-FT}
\E(x)=w^*-\lim_{\lambda\to 0}\lambda\int_0^\infty e^{-\lambda t}\T_t(x)\,dt
=w^*-\lim_{t\to +\infty}\frac{1}{t}\int_0^t\T_s(x)\,ds.
 \end{equation}
Hence, $\FT$ is atomic by Theorem \ref{Tomi}, the range of the predual operator $\E_*$
coincides with $\mbox{}\,^\perp{\Ker}\E$ and it is isomorphic to $\FT_*$ through the map
\[
\Ran\E_*=\,^\perp{\Ker}\E\ni\sigma\mapsto\sigma\circ\E\in\FT_*.
\]
Moreover, we clearly have $\Ran\E_*=\mathcal{F}(\T_*)$, (see also Corollary 2.2 in \cite{Frigerio-staz}).
\end{proof}

 {Therefore, assuming the existence of a faithful invariant state, we can find a countable set $J$}, and two sequences $(\Ss_j)_{j\in J}$,
$(\ff_j)_{j\in J}$ of {separable} Hilbert spaces such that
\begin{align}\label{dec-h-FT}
&\h\simeq\oplus_{j\in J}\left(\Ss_j\otimes\ff_j\right)\qquad\qquad\quad\ \mbox{(unitary equivalence)}\\ \label{dec-FT}
&\FT\simeq\oplus_{j\in J}\left(\mathcal{B}(\Ss_j)\otimes\unit_{\ff_j}\right),\quad
\mbox{($*$-isomorphism isometric)}
\end{align}
where $\unit_{\ff_j}$ denote the identity operator on $\ff_j$.\\
{Even if the decomposition \eqref{dec-FT} is given up to an isometric isomorphism, for sake of simplicity we will identify $\h$ with $\oplus_{j\in J}\left(\Ss_j\otimes\ff_j\right)$ and $\FT$ with $\oplus_{j\in J}\left(\mathcal{B}(\Ss_j)\otimes\unit_{\ff_j}\right)$.}
\smallskip\\

Now we can {state} for $\FT$ a similar result to Theorem \ref{th:main}. {Note that, it has already been proved for a quantum channel on a matrix algebra in \cite{viola} Lemma $6$, and in \cite{wolf} Theorems $6.12$ and $6.14$. Here, we extend this in the infinite-dimensional framework.}

\begin{theorem}\label{th:dec-FT}
Assume there exists a faithful normal invariant state. Let $(\Ss_j)_{j\in J}$ and $(\ff_j)_{j\in J}$ be two countable sequences of Hilbert spaces such that
\eqref{dec-h-FT} and \eqref{dec-FT} hold.
Then we have the following facts:
\begin{enumerate}
\item for every GKSL representation (\ref{eq:GKSL}) of the generator $\Ll$ by means
of operators $L_\ell, H$, we have
\begin{align*}
L_\ell  & = \oplus_{j\in J}\left(\unit_{\Ss_j}\otimes N^{(j)}_\ell\right)\quad\forall\,\ell\ge 1,\\
H & = \oplus_{j\in J}
\left(\lambda_j\unit_{\Ss_j}\otimes \unit_{\ff_j} + \unit_{\Ss_j}\otimes N^{(j)}_0\right),
\end{align*}
where $N^{(j)}_\ell$ are operators on $\ff_j$ such that the series
$\sum_{\ell}(N^{(j)}_\ell)^* N^{(j)}_\ell$ are strongly convergent for all $j\in J$, $(\lambda_j)_{j\in J}$ is a sequence of real numbers, and  every
 $M^{(j)}_0$ is a self-adjoint operator on $\ff_j$;
 \item $\T_t(x\otimes y)=x\otimes\T^{\ff_j}_t(y)$ for all $x\in\mathcal{B}(\Ss_j)$ and $y\in\mathcal{B}(\ff_j)$, for $j\in J$, where $\T^{\ff_j}$ is the QMS on $\mathcal{B}(\ff_j)$ generated by  $\Ll^{\ff_j}$, {whose GKSL representation is given by $\{
N^{(j)}_\ell, N^{(j)}_0\,:\ell\geq 1\}$;
}
\item every $\T^{\ff_j}$ is irreducible and possesses a unique (faithful) normal
invariant state $\tau_{\ff_j}$;
\item every invariant state $\eta$ has the form $\eta=\sum_{j\in J}\sigma_j\otimes\tau_{\ff_j}$
with $\sigma_j$ an arbitrary {positive} trace-class operator on $\Ss_j$ {such that $\sum_{j\in J}\tr{\sigma_j}=1$}.
\end{enumerate}
\end{theorem}
\begin{proof}
Since \eqref{dec-FT} holds, like in the proof of Theorems 3.1, 3.2 in \cite{DFSU}
there exist operators $(N^{(j)}_\ell)_\ell$ on $\mathcal{B}(\ff_j)$
such that $L_\ell  = \oplus_{j\in J}\left(\unit_{\Ss_j}\otimes N^{(j)}_\ell\right)$ for all
$\ell\geq 1$ and $j\in J$.\\
Now, if $p_j$ is the orthogonal projection onto $\Ss_j\otimes\ff_j$, we have $H=\sum_{l,m}p_lH_{lm}p_m$ with
$H_{lm}:\h_{\Ss_m}\otimes\h_{\ff_m}\to\h_{\Ss_l}\otimes\h_{\ff_l}$ and $H_{lm}^*=H_{ml}$
for all $l,m\in J$. Since every $x=\oplus_{j\in J}(x_j\otimes\unit_{\ff_j})\in\FT$ commutes
with $H$ we get $0=[x,H]$, i.e.
\[
0=\sum_jp_j[x_j\otimes\unit_{\ff_j},H_{jj}]p_j+\sum_{j\neq m}p_j\left((x_j\otimes\unit_{\ff_j})H_{jm}-H_{jm}(x_m\otimes\unit_{\ff_m})\right)p_m,
\]
which implies
\[
[x_j\otimes\unit_{\ff_j},H_{jj}]=0\quad\forall\,j\in J,\quad (x_j\otimes\unit_{\ff_j})H_{jm}=H_{jm}(x_m\otimes\unit_{\ff_m})\quad\forall\,j\neq m.
\]
The first condition is equivalent to have $H_{jj}=\lambda_j\unit_{\Ss_j}\otimes \unit_{\ff_j}
+ \unit_{\Ss_j}\otimes N^{(j)}_0$ for some $ N^{(j)}_0\in\mathcal{B}(\ff_j)$ and $\lambda_j\in\mathbb{R}$; the second one gives $H_{jm}=0$ for all $j\neq m$, and so we obtain
item 1.\\
Item 2 trivially follows. The proof of items 3 and 4 are similar to the ones of
Theorem 4.1 and 4.3, respectively, in \cite{DFSU}.
\end{proof}

We want now to understand the relationships between decompositions \eqref{eq:UnT} and \eqref{dec-FT} {making use of the notations introduced in Theorems \ref{th:main} and \ref{th:dec-FT}. In particular, in Theorem \ref{NT-atomic-FT}, we find a spectral characterization of the decomposition of the fixed point algebra, up to an isometric isomorphism. Indeed, in this representation, the spaces $\Ss_j$ undergoing trivial evolutions are the eigenspaces of suitable Hamiltonians $K_i$ corresponding to their different eigenvalues.
\smallskip

First of all we introduce the following notation: for every $i\in I$ denote by
\begin{equation}\label{def-spc-Ki}
\sigma(K_i):=\{\kappa_j^{(i)}\,: j\in J_i\}
\end{equation}
with $\kappa_j^{(i)}\ne\kappa_l^{(i)}$ for $j\neq l$ in $J_i$, the (pure point) spectrum of
the Hamiltonian $K_i\in\mathcal{B}(\kk_i)$ for some at most countable set $J_i\subseteq\mathbb{N}$. Note that, if $\T$ has a faithful normal invariant state, then $\sigma(K_i)$ is exactly the spectrum of $K_i$ thanks to Theorem \ref{th:main}.\\
Without of loss of generality we can choose the family $\{J_i\,:\, i\in I\}$ such that $J_h\cap J_l=\emptyset$ whenever $h\neq l$. In this way, set
\begin{equation}\label{def-J} J:=\cup_{i\in I}J_i,
\end{equation}
for $j\in J$ there exists a unique $i\in I$ such that $j=j_i\in J_i$.

\begin{theorem}\label{NT-atomic-FT}
Assume $\nT$ atomic and let $\nT=\oplus_{i\in I}\left(\mathcal{B}(\kk_i)\otimes\unit_{\mm_i}\right)$ with $(\kk_i)_i$, $(\mm_i)_i$ two countable sequences of Hilbert spaces such that
$\h=\oplus_{i\in I}(\kk_i\otimes \mm_i)$. If there exists a faithful normal invariant state, up to an isometric isomorphism we have
\begin{equation}\label{rappr-FT}\FT=\oplus_{j\in J}\left(\mathcal{B}(\Ss_j)\otimes\unit_{\ff_j}\right)\end{equation} with
$J$ defined in \eqref{def-J}, and
\begin{equation}\label{def-sj-fj}
\Ss_j=\Ss_{j_i}:={\Ker}\left(K_i-\kappa_j^{(i)}\unit_{\kk_i}\right),\qquad
\ff_j=\ff_{j_i}:=\mm_i\qquad\forall\,j_i\in J_i,\ i\in I.
\end{equation}
\end{theorem}
\begin{proof} By considering the spectral decomposition $K_i=\sum_{j\in J_i}\kappa_j^{(i)}q_{ji}$
with $(q_{ji})_{j\in J_i}$ mutually orthogonal projections such that
$$q_{ji}\kk_i=\Ker \left(K_i-\kappa_j^{(i)}\unit_{\kk_i}\right)=:\Ss_{j_i}$$ and
$\sum_{j\in J_i}q_{ji}=\unit_{\mm_i}$, we immediately obtain
\[
\kk_i\otimes\mm_i=\left(\oplus_{j\in J_i}\Ss_{j_i}\right)\otimes\mm_i=\oplus_{j\in J_i}\left(\Ss_{j_i}\otimes\ff_{j_i}\right)
\]
by setting $\ff_{j_i}:=\mm_i$ for all $j\in J_i$. Therefore, by definition of $J$, since every $j\in J$ belongs to a unique $J_i$, we have
\[
\h=\oplus_{i\in I}(\kk_i\otimes \mm_i)=\oplus_{i\in I}\oplus_{j\in J_i}\left(\Ss_{j_i}\otimes\ff_{j_i}\right)=\oplus_{j\in J}\left(\Ss_j\otimes\ff_j\right).
\]

In order to conclude the proof we have to show {equality \eqref{rappr-FT}}.
Given {$x\in\FT\subseteq\nT$ (see item $4$ in Proposition \ref{prop-struct-NT})}, we can write $x=\oplus_{i\in I}(x_i\otimes \unit_{\mm_i})$ with
$(x_i)_{i\in I}\subseteq\mathcal{B}(\kk_i)$, and so, by Theorem \ref{th:main} we have
\[
x=\oplus_{i\in I}(x_i\otimes \unit_{\mm_i})
=\T_t(\oplus_{i\in I}(x_i\otimes \unit_{\mm_i}))
=\oplus_{i\in I}(e^{\mi tK_i}x_ie^{-\mi tK_i}\otimes \unit_{\mm_i}).
\]
Consequently, $x_i=e^{\mi tK_i}x_ie^{-\mi tK_i}$ for all $i\in I$, i.e. every $x_i$
commutes with $K_i$, and then with each projection $q_{ji}$ with $j\in J_i$. This means that
each $x_i=\oplus_{j\in J_i}q_{ji}x_iq_{ji}$ belongs to the algebra
$\oplus_{j\in J_i}q_{ji}\mathcal{B}(\kk_i)q_{ji}
=\oplus_{j\in J_i}\mathcal{B}(q_{ji}\kk_i)
=\oplus_{j\in J_i}\mathcal{B}(\Ss_j)$, so that $x$ is in
$\oplus_{j\in J}\left(\mathcal{B}(\Ss_j)\otimes\unit_{\ff_j}\right)$.

On the other hand, given $i\in I$ and $j\in J_i$, for $u,v\in\Ss_j={\Ker}(K_i-\kappa_j^{(i)}\unit_{\kk_i})$ we get
$$
\T_t(\kb{u}{v}\otimes\unit_{\mm_i})=\kb{e^{\mi tK_i}u}{e^{\mi tK_i}v}\otimes\unit_{\mm_i}=\kb{u}{v}\otimes\unit_{\mm_i}\qquad\forall\,t\geq 0,
$$
where $u$ and $v$ are eigenvectors of $K_i$ associated with the same eigenvalue $\kappa_j^{(i)}$.
Since ${\Ker}(K_i-\kappa_j^{(i)}\unit_{\kk_i})$ is generated by
elements of the form $\kb{u}{v}$, and the net $(\T_t(z))_t$ is uniformly bounded for all $z$,
we obtain the inclusion $\mathcal{B}(\Ss_j)\otimes\unit_{\mm_i}\subseteq \FT$ for all $j\in J_i$, $i\in I$, i.e. $\mathcal{B}(\Ss_j)\otimes\unit_{\ff_j}\subseteq\FT$ for all $j\in J=\cup_{i\in I}J_i$.
\end{proof}
Theorem below shows how we can derive an \lq\lq atomic decomposition \rq\rq of $\FT$ from one of $\nT$'s. We want now to analyze the opposite procedure.
}

Assuming $\FT=\oplus_{j\in J}\left(\mathcal{B}(\Ss_j)\otimes\unit_{\ff_j}\right)$
with $(\Ss_j)_j$, $(\ff_j)_j$ two countable sequences of Hilbert spaces such that
$\h= \oplus_{j\in j}(\Ss_j\otimes \ff_j)$, and using notations of Theorem \ref{th:dec-FT},
we set an equivalence relation on $J$ in the following way:
\begin{definition} Given $j,k\in J$, we say that \emph{$j$ is in relation with $k$} (and write \emph{$j\sim k$}) if there exist a complex separable Hilbert space $\mm$ and unitary isomorphisms
\begin{equation}\label{V_j}V_j:\ff_j\to\mm,\qquad V_k:{\ff_k}\to\mm\end{equation} such that operators $\{V_jN_l^{(j)}V_j^*, V_jN_0^{(j)}V_j^*\,:\,l\geq 1\}$ and $\{V_kN_l^{(k)}V_k^*, V_kN_0^{(k)}V_k^*\,:\,l\geq 1\}$ give the same Lindbladian operator on $\mathcal{B}(\mm)$.
\end{definition}
We obtain in this way an equivalence relation which induces a partition of $J$, $J=\cup_{ n\in I}I_n$, for some finite or countable set $I\subseteq\mathbb{N}$, where each $I_n$ is an equivalence class with respect to $\sim$.

\begin{theorem}\label{FT-atomic-NT}
Assume that there exists a faithful normal invariant state and $\nT$ atomic.
Let $\FT=\oplus_{j\in J}\left(\mathcal{B}(\Ss_j)\otimes\unit_{\ff_j}\right)$
with $(\Ss_j)_j$, $(\ff_j)_j$ two countable sequences of Hilbert spaces such
that $\h= \oplus_{j\in j}(\Ss_j\otimes \ff_j)$, and let $\{I_n\,:\,n\in I\}$
be the set of equivalence classes of $J$ with respect to the relation $\sim$. Then $\nT$
is isometrically isomorphic to the direct sum $\oplus_{n\in I}\left(\mathcal{B}(\kk_n)\otimes\unit_{\mm_n}\right)$ with
\begin{equation}\label{def-kk-mm}\kk_n:=\oplus_{j\in I_n}\Ss_j,\qquad\mm_n:=V_j\ff_j\quad\forall\,j\in I_n,\end{equation}
where $V_j$'s are the unitary isomorphisms given in \eqref{V_j}.
\end{theorem}
\begin{proof} Given $n\in I$, by definition of $\mm_n$ we can define a unitary operator by setting
$$\begin{array}{rcl}U_n:\oplus_{j\in I_n}\left(\Ss_j\otimes\ff_j\right)&\to&\left(\oplus_{j\in I_n}\Ss_j\right)\otimes\mm_n=\kk_n\otimes\mm_n\\
\oplus_{j\in I_n}\left(u_j\otimes z_j\right)&\mapsto&\sum_{j\in I_n}\left(u_j^{(j)}\otimes V_jz_l\right)
\end{array}
$$
where $ u_j^{(j)}$ in $\oplus_{i\in I_n}\h_{\Ss_i}$ denotes the vector $$u_j^{(j)}:=\oplus_{i\in I_n}v_i,\qquad v_i:=\left\{\begin{array}{ll}0 &\ \mbox{if $i\neq j$}\\
u_j&\ \mbox{if $i=j$}\end{array}\right..$$
Now, since $$\h= \oplus_{j\in j}(\Ss_j\otimes \ff_j)=\oplus_{n\in I}\left(\oplus_{j\in I_n}(\Ss_j\otimes \ff_j)\right)$$ by the equality $J=\cup_{n\in I}I_n$,
by setting $U:=\oplus_{n\in I}U_n$ we get a unitary operator $U:\h\to\oplus_{n\in I}\left(\kk_n\otimes\mm_n\right)$ such that
$$U\FT U^*=\oplus_{n\in I}\left(\left(\oplus_{j\in I_n}\mathcal{B}(\Ss_j)\right)\otimes\unit_{\mm_n}\right).$$
In order to conclude the proof we have to show that $$U\nT U^*=\oplus_{n\in I}\left(\mathcal{B}(\oplus_{j\in I_n}\Ss_j)\otimes\unit_{\mm_n}\right).$$
To this end recall that, by Theorem \ref{th:dec-FT}, the operators $(L_\ell)_\ell,\ H$ in a GKSL representation of the generator $\Ll$ can be written as
\[
L_\ell  = \oplus_{j\in J}\left(\unit_{\Ss_j}\otimes N^{(j)}_\ell\right) \quad\forall\,\ell\ge 1,
\qquad H = \oplus_{j\in J}\left(\lambda_j\unit_{\Ss_j}\otimes \unit_{\ff_j} + \unit_{\Ss_j}\otimes N^{(j)}_0\right)
\]
with $(N_\ell^{(j)})_{\ell\geq 1}\subseteq\mathcal{B}(\Ss_j)$, $N^{(j)}_0=(N^{(j)}_0)^*\in\mathcal{B}(\ff_j)$ and $\lambda_j\in\mathbb{R}$ for all $j\in J$; moreover, by definition of $I_n$, we can choose operators $(M_l^{(n)})_l$ and $M_0^{(n)}=(M_0^{(n)})^*$ in $\mathcal{B}(\mm_n)$ such that $\{M_l^{(n)}, M_0^{(n)}\,:\,l\geq 1\}$ is a GKSL representation of the  the generator $\Ll^{\mm_n}$ of a QMS $\T^{\mm_n}$ on $\mathcal{B}(\mm_n)$ equivalent to $\{V_jN_l^{(j)}V_j^*, V_jN_0^{(j)}V_j^*\,:\,l\geq 1\}$ for {all} $j\in I_n$. Therefore, the operators
\begin{align*}L_\ell^\prime&:=\oplus_{n\in I}\left(\oplus_{j\in I_n}\left(\unit_{\Ss_j}\otimes V_j^*M^{(n)}_\ell V_j\right)\right)\\
H^\prime&:=\oplus_{n\in I}\left(\oplus_{j\in I_n}\left(\lambda_j\unit_{\Ss_j}\otimes \unit_{\ff_j} \right)+ \oplus_{j\in I_n}\left(\unit_{\Ss_j}\otimes V_j^*M^{(n)}_0V_j\right)\right),
\end{align*}
clearly give the same GKSL representation of $\{L_\ell,\ H\,:\,\ell\geq 1\}$.
Moreover we have
\begin{align*}UL_\ell^\prime U^*&=\oplus_{n\in I}U_n\left(\oplus_{j\in I_n}\left(\unit_{\Ss_j}\otimes V_j^*M^{(n)}_\ell V_j\right)\right)U_n^*=\oplus_{n\in I}\left(\unit_{\kk_n}\otimes M^{(n)}_\ell\right)\\
UH^\prime U^*&=\oplus_{n\in I}U_n\left(\oplus_{j\in I_n}\left(\lambda_j\unit_{\Ss_j}\otimes \unit_{\ff_j} \right)+ \oplus_{j\in I_n}\left(\unit_{\Ss_j}\otimes V_j^*M^{(n)}_0V_j\right)\right)U_n^*\\
&=\oplus_{n\in I}\left(K_n\otimes \unit_{\mm_n} + \unit_{\kk_n}\otimes M^{(n)}_0\right)
\end{align*}
with $K_n:=\left(\oplus_{j\in I_n}\lambda_j\unit_{\Ss_j}\right)=K_n^*\in\mathcal{B}(\kk_n)$, so that
\begin{align*}U\delta_{H\prime}^m(L_\ell^\prime)U^*=\delta_{UH^\prime U^*}^{m}(UL_\ell^\prime U^*)&=\oplus_{n\in I}\left(\unit_{\kk_n}\otimes\delta_{M_0^{(n)}}^{m}(M_\ell^{(n)})\right)\\
U\delta_{H\prime}^m(L_\ell^{\prime *})U^*=\delta_{UH^\prime U^*}^{m}(UL_\ell^{\prime *} U^*)&=\oplus_{n\in I}\left(\unit_{\kk_n}\otimes\delta_{M_0^{(n)}}^{m}(M_\ell^{(n)*})\right)
\end{align*}
for all $m\geq 0$. Therefore, since $U\nT U^*=\left\{U\delta_{H\prime}^m(L_\ell^\prime)U^*,U\delta_{H\prime}^m(L_\ell^{\prime *})U^*\,:\,m\geq 0\right\}^\prime$ by item $2$ of Proposition \ref{prop-struct-NT}, we obtain that an operator $x\in\mathcal{B}\left(\oplus_{n\in I}\left(\kk_n\otimes\mm_n\right)\right)$ belongs to $U\nT U^*$ if and only if it commutes with $$\oplus_{n\in I}\left(\unit_{\kk_n}\otimes\delta_{M_0^{(n)}}^{m}(M_\ell^{(n)})\right),\qquad\oplus_{n\in I}\left(\unit_{\kk_n}\otimes\delta_{M_0^{(n)}}^{m}(M_\ell^{(n)*})\right)$$ for every $m\geq 0$.

Now, let $q_n:=\unit_{\kk_n}\otimes\unit_{\mm_n}$ the orthogonal projection onto $\kk_n\otimes\mm_n$. We clearly have $q_n\in U\FT U^*$ and $\sum_{n\in I}q_n=\unit$. Therefore, the algebra $q_n\B q_n=\mathcal{B}(\kk_n\otimes\mm_n)$
is preserved by the action of every map $\wT_t:=U\T_t(U^*\cdot U)U^*$, and so we can
consider the restriction of $\wT$ to this algebra, getting a QMS
on $\mathcal{B}(\kk_n\otimes\mm_n)$ denoted by $\T^{(n)}$ and satisfying $\mathcal{N}(\T^{(n)})=q_nU\nT U^*q_n$, where $\mathcal{N}(\T^{(n)})$ is the decoherence-free algebra of $\T^{(n)}$. Note that, since each $q_n$ commutes with every
$UL_\ell^\prime U^*$, $UL_\ell^{\prime *} U^*$ and $UH^\prime U^*$, a GKSL representation of the generator $\Ll^{(n)}$ of $\T^{(n)}$ is given by operators $$q_nUH^\prime U^*q_n=K_n\otimes\unit_{\mm_n}+\unit_{\kk_n}\otimes M_0^{(n)}\qquad q_nUL^\prime_\ell U^*q_n=\unit_{\kk_n}\otimes M_\ell^{(n)},$$ so that $$\wT_t(x\otimes y)=e^{\mi tK_n} xe^{-\mi t K_n}\otimes\T_t^{\mm_n}(y)\qquad\forall\,x\in\mathcal{B}(\kk_n),\ y\in\mathcal{B}(\mm_n)$$ and
\begin{align}\nonumber\mathcal{N}(\T^{(n)})&=\{\unit_{\kk_n}\otimes\delta_{M_0^{(n)}}^{m}(M_\ell^{(n)}), \unit_{\kk_n}\otimes\delta_{M_0^{(n)}}^{m}(M_\ell^{(n)*})\,:\ell\geq 1,\ m\geq 0\}^\prime\\
\label{NTn-tens}&=\mathcal{B}(\kk_n)\otimes\mathcal{N}(\T^{\mm_n}).
\end{align}
%

Since $$U\nT U^*=\bigoplus_{n,m\in I}q_nU\nT U^*q_m$$ and equation \eqref{NTn-tens} holds, to conclude the proof we have to show that $$q_nU\nT U^*q_m=\{0\}\quad\forall\,n\neq m,\qquad\mbox{and}\qquad\mathcal{N}(\T^{\mm_n})=\mathbb{C}\unit_{\mm_n}.$$ So, let $x\in U\nT U^*$ and consider $n,m\in I$ with $n\neq m$. Since the net $(\T_t(q_nxq_m))_{t\geq 0}$ is bounded in norm and the unit ball is weakly* compact, there exists a weak* cluster point $y$ such that $y=w^*-\lim_\alpha\T_{t_\alpha}(q_nxq_m)$. Therefore, for $\sigma\in\RT$ with $\T_{*t}(\sigma)=e^{\mi t\lambda}\sigma$ for some $\lambda\in\mathbb{R}$, we have
$$\tr{\sigma y}=\lim_\alpha\tr{\sigma\T_{t_\alpha}(q_nxq_m)}=\lim_\alpha e^{\mi t_\alpha\lambda}\tr{\sigma q_nxq_m}.$$ Now, if $\lambda\neq 0$ this implies $\tr{\sigma q_nxq_m}=0=\tr{\sigma y}$; otherwise, since $\sigma$ is an invariant state, we automatically have $q_m\sigma q_n=0$ for $n\neq m$ by Theorem \ref{th:dec-FT}, so that $\tr{\sigma y}$ is $0$ again. Consequently, $\tr{\sigma y}=0$ for all $\sigma\in\RT$. On the other hand, taking $\eta$ such that $0\in\overline{\{\T_{*t}(\eta)\}}^w_{t\geq 0}$, up to passing to generalized subsequences we have
$$\tr{\eta y}=\lim_\alpha\tr{\eta\T_{t_\alpha}(q_nxq_m)}=\lim_\alpha \tr{\T_{*t_\alpha}(\eta)q_nxq_m}=0.$$
We can then conclude that $\tr{\sigma y}=0$ for all $\sigma\in\mathfrak{I}(\oplus_n\left(\kk_n\otimes\mm_n\right))$ by virtue of equation \eqref{dec-J}, and so $y=0$. This means that $q_nxq_m$ belongs to $\nT\cap\mathfrak{M}_s$ and this intersection is $\{0\}$ by Corollary \ref{equiv-NT-atom} ($\nT$ is atomic). Therefore $q_nxq_m=0$ for $n\neq m$ and so $$U\nT U^*=\oplus_{n\in I}q_n\nT q_n=\oplus_{n\in I}\mathcal{N}(\T^{(n)}).$$
Moreover, since $\nT$ is atomic by Corollary \ref{equiv-NT-atom}, {the general theory of von Neumann algebras (see e.g. \cite{Take}) says that} each algebra $\mathcal{N}(\T^{(n)})=\mathcal{B}(\kk_n)\otimes\mathcal{N}(\T^{\mm_n})$ is atomic too, and consequently $\mathcal{N}(\T^{\mm_n})$ is atomic. Finally, recalling that, by construction, $\T_t^{\mm_n}=\T_t^{\ff_j}$ for all $j\in I_n$ with $\T_t^{\ff_j}$ an irreducible QMS having a faithful invariant state (see Theorem \ref{th:dec-FT}), we get $\mathcal{N}(\T^{\mm_n})=\mathbb{C}\unit_{\mm_n}$ thanks to Proposition $4.3$ in \cite{DFSU}.
\end{proof}

\begin{remark}\rm
Theorem \ref{FT-atomic-NT} provides in particular a way to pass from the decomposition of GKSL operators $H$, $(L_l)_l$ of $\Ll$ according to the splitting of $\h=\oplus_{j\in j}(\Ss_j\otimes \ff_j)$ associated with  the atomic algebra $\FT$, to the other one decomposition with respect to the splitting $\h=\oplus_{n\in I}\left(\kk_n\otimes\mm_n\right)$ associated with $\nT$.

More precisely, if $\FT=\oplus_{j\in J}\left(\mathcal{B}(\Ss_j)\otimes\unit_{\ff_j}\right)$ and
\[
L_\ell  = \oplus_{j\in J}\left(\unit_{\Ss_j}\otimes N^{(j)}_\ell\right) \quad\forall\,\ell\ge 1,
\qquad H = \oplus_{j\in J}\left(\lambda_j\unit_{\Ss_j}\otimes \unit_{\ff_j} + \unit_{\Ss_j}\otimes N^{(j)}_0\right)
\]
with $(N_\ell^{(j)})_{\ell\geq 1}\subseteq\mathcal{B}(\Ss_j)$, $N^{(j)}_0\in\mathcal{B}(\ff_j)$ and $\lambda_j\in\mathbb{R}$ for all $j\in J$, then we can decompose $H, (L_\ell)_\ell$ with respect to the splitting $\h=\oplus_{n\in I}\left(\kk_n\otimes\mm_n\right)$ given by \eqref{def-kk-mm}
as follows:
\[
L_\ell=\oplus_{n\in I}
\left(\unit_{\kk_n}\otimes M_\ell^{(n)}\right),\qquad H=\oplus_{n\in I}\left(K_n\otimes\unit_{\mm_n}
+\unit_{\kk_n}\otimes M_0^{(n)}\right),
\]
where
$(M_\ell^{(n)})_{\ell\geq 1}, M^{(n)}_0=M^{(n)*}_0$ are operators in $\mathcal{B}(\mm_n)$ giving the same GKSL representation of $(N_l^{(j)})_{\ell\geq 1}, N_0^{(j)}$ for all $j\in I_n$, and every $K_n$ is the self-adjoint operator on  $\mathcal{B}({\kk_n})$ having $(\lambda_j)_{j\in I_n}$ as eigenvalues and $(\Ss_j)_{j\in I_n}$ as corresponding eigenspaces.
\end{remark}

\smallskip

{\bf Acknowledgements}. The financial support of MIUR FIRB 2010 project RBFR10COAQ
``Quantum Markov Semigroups and their Empirical Estimation'' is gratefully acknowledged.

\null\bigskip
\begin{quote}
\footnotesize{$^1$ Department of Mathematics, Politecnico di Milano,
Piazza Leonardo da Vinci 32,  I - 20133 Milano, Italy
  E-mail: \texttt{franco.fagnola@polimi.it} \\
$^2$ Department of Mathematics, University of Genova,
Via Dodecaneso,\\  I - 16146 Genova, Italy
E-mail: \texttt{sasso@dima.unige.it} \\
$^3$ Department of Mathematics, University of Genova,
Via Dodecaneso,\\  I - 16146 Genova, Italy
E-mail: \texttt{umanita@dima.unige.it}
}
\end{quote}

\end{document}